\documentclass[11pt]{article}
\usepackage{amsmath,amsfonts,amsthm,amssymb}

\newcommand{\eq}{\begin{equation}}
\newcommand{\en}{\end{equation}}
\newcommand{\eqa}{\begin{eqnarray}}
\newcommand{\ena}{\end{eqnarray}}

\newcommand{\calF}{\mathcal{F}}
\newcommand{\F}{\mathbb{F}}
\newcommand{\Z}{\mathbb{Z}}

\newcommand{\N}{\mathbb{N}}
\newcommand{\C}{\mathbb{C}}

\newcommand{\T}{\mathcal{T}}
\newcommand{\U}{\mathcal{U}}
\newcommand{\CO}{\mathcal{O}}

\numberwithin{equation}{section}

\newtheorem{corollary}[equation]{Corollary}
\newtheorem{lemma}[equation]{Lemma}

\theoremstyle{definition}

\begin{document}

\setlength{\unitlength}{1mm} \thispagestyle{empty}

%\begin{flushright}
%Draft for Internal Circulations\\
% v1: 11/11/08, v2: 1/10/09,\\
%v3: 2/28/09, v4: 3/31/09
%\end{flushright}

 %\vspace*{0.1cm}

 \begin{center}
 {\bf Quantum Fourier Transform Over Galois Rings
 \footnote{This project was initiated when Yong Zhang was visiting
 Pawel Wocjan at the School of Electrical Engineering and Computer
 Science in University of Central Florida. The author thanks the local
 host for his relevant discussions. } }

  \vspace{.3cm}

  {Yong Zhang \\[.5cm]
 Center for High Energy Physics, Peking University\\
 Beijing,  100871, P. R. China

 }

 \end{center}

 \vspace{0.2cm}

\begin{center}
\parbox{14cm}{
\centerline{\small  \bf Abstract}  \noindent\\

  Galois rings are regarded as ``building blocks" of a finite
  commutative ring with identity. There have been many papers
  on classical error correction codes over Galois rings published.
  As an important warm-up before exploring
  quantum algorithms and quantum error correction codes over
  Galois rings, we study the quantum Fourier transform (QFT) over
  Galois rings and prove it can be efficiently preformed on
  a quantum computer. The properties of the QFT over
  Galois rings lead to the quantum algorithm for
  hidden linear structures over Galois rings.

 }

\end{center}

\vspace{.2cm}

\begin{tabbing}
Quantum Error Correction Codes
           \\[.2cm]

\end{tabbing}

%----------------------
 \newpage

 \section{Introduction}

  Quantum Fourier transform (QFT) is a main tool in constructing
  some quantum algorithms, for example, Shor's algorithm \cite{shor94}.
  Readers are invited to refer to the Chapter 5 in the textbook by
  Nielsen and Chuang \cite{nc99} for the definition of the QFT and
  its applications. The QFT over finite fields was introduced in
  two papers \cite{bcw00, dh00}: the one by De Beaudrap, Cleve and
  Watrous and the other by Van Dam and Hallgren. The properties of
  the QFT over finite fields directly give rise to quantum algorithms
  for hidden linear/non-linear structures over finite fields
  \cite{bcw00, ddw07}.

  A ring contains more algebraic structures than a field: every field is
  a ring but not every ring is a field.  Galois rings \cite{macdonald74, bf02, wan03}
  are regarded as ``building blocks" of a ring. Quantum
  information and computation \cite{nc99} over Galois rings
  is to be a very meaningful research topic. The QFT over Galois rings
  possibly leads to interesting quantum algorithms, for examples,
  quantum algorithms for hidden linear/non-linear structures over
  Galois rings (or even a ring).
  Quantum error correction codes \cite{crss98} over Galois rings can
  be explored because there are many classical error
  correction codes over Galois rings, see \cite{bf02, wan03, hkcss02}
  for relevant references. Hence our research on the QFT over Galois rings is
  an important warm-up to study quantum algorithms and quantum error correction
  codes over Galois rings.

  The remainder of this paper is organized as follows. Section 2 collects basic
  facts on Galois rings used in the following sections. Section 3 defines
  the QFT over Galois rings and analyzes its main properties. The second proof
  for the lemma \ref{char} comprehensively exploits various properties of Galois
  rings. Section 4 proves that the QFT over Galois rings can be efficiently
  implemented on a quantum computer. The properties of the discriminant matrix
  over Galois rings are discussed from different points of view. Last section
  remarks the QFT over a ring.

\section{Preliminary on Galois rings}

A {\em ring} $A$ is a set equipped with addition and multiplication.
It is an abelian group with the unit $0$ under addition, denoted by
$(A, +)$. It is a semigroup with the unit $1$ under multiplication,
denoted by $(A, .)$, in which an invertible element is called a {\em
unit} but some elements may have no inverses. A nonzero element $a
\in A$ is called a {\em zero divisor} if there is another nonzero
element $b\in A$ satisfying $ab=0$.

Our notations on Galois rings are taken from the book \cite{wan03}
by Wan. $R\equiv\Z_{p^s}$ denotes the residue class ring of integers
$\Z$ modulo $p^s$ for a prime number $p$ and an integer $s\ge 1$,
i.e., \eq R \equiv \Z_{p^s}=\{0, 1, 2, \cdots, p^s-1 \} \,. \en
$R^\prime\equiv GR(p^s,p^{sm})$ denotes a {\em Galois ring} of
characteristic $p^s$ and cardinality $p^{sm}$, where $m$ is some
integer $m\ge 1$. The ring $R$ is a subring of $R^\prime$, and the
ring $R^\prime$ is an extension of the ring $R$.

For $m=1$, the Galois ring $R'$ corresponds to a residue class ring
$R^\prime|_{m=1}=\Z_{p^s}$, and for $s=1$, it corresponds to a
finite field $R^\prime|_{s=1}=\F_{p^m}$.

An arbitrary element $\alpha$ of the Galois ring $R^\prime$ can be
expressed in two ways. In the {\em additive formalism}, $\alpha$ is
uniquely expressed as \eq \label{additive} \alpha = \sum_{i=0}^{m-1}
a_i \xi^i \quad\mbox{with}\quad a_i \in R\,, \en where $\xi$ is a
root of a {\em monic basic primitive polynomial}, \eq
\label{polynomial} h(X) = h_0 + h_1 X + \cdots + h_{m-1} X^{m-1} +
X^m \in R[X] \en of degree $m$ over $R$.

In the {\em $p$-adic formalism}, $\alpha$ is uniquely expressed as
\eq \label{padic} \alpha=\sum_{i=0}^{s-1} t_i \, p^i
\quad\mbox{with}\quad t_i \in \T_{p^m} =\{0,1,\xi,\cdots,
\xi^{p^m-2} \}\,, \en where the set $\T_{p^m}$ is referred to as the
{\em Teichm{\"u}ller set}.

In the $p$-adic formalism (\ref{padic}), $\alpha$ is a unit if and
only if $t_0\neq 0$, and it is a zero divisor or $0$ if and only if
$t_0=0$.

\begin{lemma}
\label{zdu} Given an arbitrary zero divisor $\alpha$ of the Galois
ring $R^\prime$, it can be expressed as  $\alpha=p^j \alpha^\prime$,
$1\le j \le s-1$ where $\alpha^\prime$ is a unit of $R^\prime$.
\end{lemma}

\begin{proof}
In the $p$-adic formalism, the zero divisor $\alpha\in R^\prime$ has
a unique form $$
 \alpha=t_j p^j + \cdots + t_{s-1} p^{s-1} =p^j \alpha^\prime, \quad
 t_j \neq 0, \quad 1 \le j \le s-1
$$
where $\alpha^\prime=t_j+ \cdots + t_{s-1} p^{s-j-1}$ is a unit of
$R^\prime$.
\end{proof}

The {\em  Frobenius automorphism} of the Galois ring $R^\prime$ over
$R$ is a map $\phi$ uniquely defined by $\phi(\xi)=\xi^p$ and
$\phi(r)=r$ for all $r\in R$.   Define a composition of the
Frobenius automorphism $\phi$ as \eq \phi^{i+1} =\phi^i \circ \phi,
\quad \phi^0=1, \quad i=0, 1, 2, \ldots\,, m - 1\,. \en Observe that
$\phi^m=1$.  The Galois group ${\rm Gal}(R'/R)$ of the ring
extension $R'/R$ is the cyclic group $\langle\phi\rangle$ of order
$m$ generated by $\phi$.  The trace ${\rm Tr} \equiv {\rm
Tr}_{R'/R}$ of this ring extension is defined by \eq \label{trace}
Tr\, :\,  R^\prime \to R\,, \quad Tr(\alpha) = \sum_{\varphi\in {\rm
Gal}(R'/R)} \varphi(\alpha)\,. \en It satisfies the following
properties:
 \eqa && Tr(\alpha + \beta)  =Tr(\alpha) + Tr(\beta), \quad
 Tr(\phi(\alpha))=Tr(\alpha), \nonumber \\
&& Tr(a\alpha)= a Tr(\alpha), \quad Tr(a) = m \, a \mbox{  for all }
a \in R\,. \ena Hence the trace mapping $Tr$ is a surjective
homomorphism from the additive group $(R^\prime, +)$ to  the
additive group $(R, +)$.

\section{The QFT over the Galois ring $R^\prime$}

We introduce the QFT over the Galois ring $R^\prime$, study its main
properties, and present two types of proofs for the lemma
\ref{char}.

\subsection{Notations}

The Galois ring $R^\prime$ is a module over its subring $R$, namely
$R\times R^\prime \to R^\prime$, see the book \cite{shoup05} by
Shoup for our notations on module and matrix over a ring.

In the additive formalism (\ref{additive}) of the Galois ring
$R^\prime$, the set $\{\xi^i\}^{m-1}_{i=0}$ forms a basis of this
module on $R$, and its element $x \in R^\prime$ has a simpler
notation
 $x= \vec{x}^T \cdot  \vec{\xi}$ with the row vector $ \vec{x}^T \in R^{1\times m}$
 and the column vector $ \vec{\xi}\in {R^\prime\,}^{m\times 1}$,
 \eq
\vec{x}^T= (x_0, \cdots, x_{m-1}), \quad
\vec{\xi\,}^T=(\xi^0,\cdots, \xi^{m-1})
 \en
where $T$ denotes the matrix transpose.

The set $\{|x_i\rangle\}_{x_i\in R}$ of all Dirac kets  is an
orthonormal basis of the Hilbert space $\C^{p^s}$, and hence the set
of all $m$-tuple tensor products of Dirac kets $|x_i\rangle$,
$i=1,\cdots, m-1$, denoted by
 $$|x\rangle \equiv |x_0\rangle\otimes |x_1\rangle\otimes \cdots \otimes
  |x_{m-1}\rangle, \quad x\in R^\prime$$
gives rise to an orthonormal basis of the Hilbert space
$(\C^{p^s})^{\otimes m}$. The set $\{|x\rangle\}$ satisfies
$\sum_{x\in R^\prime} |x\rangle\langle x|=Id_{p^s}^{\otimes m}$
where $Id_{p^s}$ denotes the $p^{s}$-dimensional identity.

Introduce a character $\chi_\alpha$ for the finite abelian group
$(R^\prime,+)$ by \eq \label{add_char}
\chi_\alpha(u)=(\omega_{p^s})^{Tr(\alpha u)},\quad \omega_{p^s}=e^{i
\frac {2\pi} {p^s}},\quad u\in R^\prime \en with the multiplication
law given by $\chi_\alpha\circ \chi_\beta\equiv
\chi_{\alpha+\beta}$, $\alpha, \beta \in R^\prime$. It satisfies
 $$\chi_\alpha(u+v)=\chi_\alpha(u)\chi_\alpha(v),\quad  u, v \in R^\prime,$$
and so it is a group homomorphism from the additive group
$(R^{\prime}, +)$ to the multiplicative semigroup $(R^\prime,
\cdot)$.

Let $\calF_{R^\prime}$ denote the QFT over the Galois ring
$R^\prime$. We use the additive character $\chi_\alpha$
(\ref{add_char}) to define $\calF_{R^\prime}$ as
 \eq
 \label{dftgalois}
\calF_{R^\prime}\equiv\frac 1 {\sqrt{p^{sm}}}
  \sum_{\alpha,\,\,u \in R^\prime} \chi_\alpha(u)| \alpha \rangle \langle
  u|.
\en At $m=1$, $\calF_{R^\prime}$ is the QFT $\calF_R$ over
$R=\Z_{p^s}$ \cite{nc99}. At $s=1$, $\calF_{R^\prime}$ is the QFT
over the finite field $\F_{p^m}$ \cite{bcw00,dh00}.

\subsection{Properties of the QFT over $R^\prime$}

We describe the properties of the $\calF_{R^\prime}$
(\ref{dftgalois}) in three corollaries of the lemma (\ref{char}).

\begin{lemma}
\label{char} Let $R^\prime$ denote the Galois ring $GR(p^s,p^{sm})$
and $\chi_\alpha$ denote the additive character, $\alpha\in
R^\prime$. Then the character $\chi_\alpha(u)$ has the property \eq
 \sum_{u\in R^\prime} \chi_\alpha(u) =p^{sm} \delta_{\alpha,0}
 =\left\{\begin{array}{cc}
  p^{sm}, & \alpha=0 \\
  0, & \alpha \neq 0
 \end{array}\right..
\en
\end{lemma}

\begin{proof}
It is obvious for $\alpha=0$.  For $\alpha\neq 0$, there exists an
element $v\in R^\prime$ satisfying $\chi_\alpha(v) \neq 1$, and we
have
$$\sum_{u\in R^\prime} \chi_\alpha(u)=\sum_ {u\in
R^\prime} \chi_\alpha(u+v)=\chi_\alpha(v)\sum_ {u\in R^\prime}
\chi_\alpha(u)\Longrightarrow \sum_ {u\in R^\prime}
\chi_\alpha(u)=0$$ since $\chi_\alpha$ (\ref{add_char}) is a
nontrivial character of the additive group $(R^\prime, +)$.
\end{proof}

\begin{corollary}
The set of all $p^{sm}$-dimensional normalized vectors,
 $$\vec{\chi}_{\alpha}=\frac 1 {\sqrt{p^{sm}}} (\chi_\alpha (u))_{u\in
 R^\prime},\quad  \alpha\in R^\prime$$ is an orthonormal basis of the
 Hilbert space $\C^{p^{sm}}$, and we have
$$\vec{\chi}_\alpha (\vec{\chi}_\beta)^\dag \equiv \frac 1 {p^{sm}}
\sum_{u\in R^\prime} \chi_\alpha(u) \chi_\beta^\ast(u)
   =\delta_{\alpha,\beta}, \quad \alpha, \beta \in R^\prime$$
where $\dag$ denotes the Hermitian conjugation.

\end{corollary}

Hence the $p^{sm}\times p^{sm}$ matrix $\frac 1 {\sqrt
{p^{sm}}}(\chi_\alpha(u))_{\alpha,u\in R^\prime}$ is a unitary
matrix, and  the QFT over $R^\prime$, $\calF_{R^\prime}$
(\ref{dftgalois}) is a unitary transformation, namely,
 $$\calF_{R^\prime} \calF^\dag_{R^\prime} =\calF^\dag_{R^\prime} \calF_{R^\prime}
 =Id_{p^s}^{\otimes m}$$
in the Hilbert space $(\C^{p^s})^{\otimes m}$.

\begin{corollary}
\label{shift} The shift operator $S_\alpha$ on the Galois ring
$R^\prime$, defined by
  $$S_\alpha \equiv \sum_{u\in R^\prime} |u+
{\,\alpha}\rangle\langle u|, \quad \alpha\in R^\prime,$$
 is diagonalized by the QFT over
$R^\prime$, $\calF_{R^\prime}$ (\ref{dftgalois}), namely,
 $$\calF_{R^\prime} S_{\alpha} \calF^\dag_{R^\prime}=\sum_{u\in R^\prime}
  \chi_\alpha(u) |u\rangle \langle u|.
 $$
\end{corollary}

\begin{proof}
After some algebra, we have \eq \calF_{R^\prime} S_{\alpha}
\calF^\dag_{R^\prime}=\frac 1 {p^{sm}} \sum_{u,\,v,\,t\in R^\prime}
\chi_\alpha(u) \chi_{u-t}(v) |  {u}\rangle \langle   {t}| \en and
then prove the corollary with the lemma \ref{char}.
\end{proof}

\begin{corollary}
\label{control} Let $A_r$ and $B_r$ denote the control additive
gates:
 $A_r|x\rangle |y\rangle \equiv |x\rangle|y+ r x\rangle$ and
 $B_r|x\rangle |y\rangle \equiv|x+r y\rangle|y\rangle$, $x,y,r\in R^\prime$.
  Then they have the
 control/target inversion property given by
 $$(\calF^\dag_{R^\prime} \otimes \calF_{R^\prime} ) A_r
  (\calF_{R^\prime} \otimes \calF^\dag_{R^\prime}) =B_r$$
\end{corollary}

\begin{proof} We use the same methodology  \cite{bcw00}
of proving the control/target inversion property for control
additive gates over finite fields. The proof is an application of
the corollary \ref{shift}.
\end{proof}

De Beaudrap and coauthors \cite{bcw00} introduced control additive
gates over finite fields and realized that  the control/target
inversion property derives the quantum algorithm for hidden linear
structures over finite fields. Therefore the corollary \ref{control}
leads to the same quantum algorithm for hidden linear structures
over Galois rings, see \cite{bcw00} for this algorithm.

\subsection{The second proof for the lemma \ref{char}}

The proof for the lemma \ref{char} is based on the fact that
$\chi_\alpha$ (\ref{add_char}) is the character of the additive
group $(R^\prime, +)$. On the other hand, the additive character
$\chi_\alpha$ (\ref{add_char}) contains information on the
multiplicative semigroup $(R^\prime, .)$, and therefore the lemma
\ref{char} can be proved only with properties of the Galois ring
$R^\prime$.

Given an element $\alpha\in R^\prime$, it is either $0$ or $1$ or a
non-identity unit or a zero divisor. Hence we prove the lemma
\ref{char} in the following four steps.

Denote $\chi(\alpha u)\equiv \chi_{\alpha}(u),\quad u\in R^\prime$.

1). $\alpha=0$. We have $\chi(0)=0$ then $\sum_{u\in R^\prime} 1
=p^{sm}$ to prove the lemma.

2). $\alpha=1$. The trace mapping $Tr$ over the Galois ring
$R^\prime$ relative to $R$ is a surjective additive group
homomorphism from $(R^\prime,+)$ to $(R,+)$. Denote the kernel of
this homomorphism by
$$ker(Tr)=\{v\in R^\prime | Tr(v)=0 \}$$
and then the quotient group $R^\prime/ker(Tr)$ is isomorphic to
$R=\Z_{p^s}$. The isomorphism gives rise to a partition of the
Galois ring $R^\prime$ as a disjoint union of the kernel $ker(Tr)$
and cosets $(z_i + Ker(T))$ with $Tr(z_i)=i$, $i=1,\cdots, p^s-1$.
This partition  derives the cardinality of $ker(Tr)$ or $(z_i +
ker(Tr))$ as $p^{(m-1)s}$. Hence we have
 \eq
 \sum_{u\in R^\prime} \chi(u) =p^{(m-1)s} \sum_{i=0}^{p^s-1}
 (\omega_{p^s})^i=0.
 \en

3). As $\alpha$ is a unit, the mapping $u \mapsto v=\alpha u$ is
bijective due to the existence of $\alpha^{-1}$, and we have
 $$\sum_{u\in R^\prime} \chi(\alpha u) = \sum_{v\in R^\prime} \chi(v)=0$$
 which uses the statement in the step 2).

 4). As $\alpha$ is a zero divisor, with the lemma \ref{zdu}, it
 has the form of $\alpha=p^j \alpha^\prime$, $1\le j \le s-1$, where
 $\alpha^\prime$ is a unit of the Galois ring $R^\prime$.
 We have
  \eq
 \sum_{u\in R^\prime} \chi(\alpha u) =\sum_{u\in R^\prime}
 (\omega_{p^{s-j}})^{Tr(\alpha^\prime u)} =\sum_{v\in R^\prime}
 (\omega_{p^{s-j}})^{Tr(v)}=0
  \en
  which exploits the steps 2) and 3).

A {\em nice} additive character for the additive group of a ring has
properties of its multiplicative semigroup so that the related QFT
over this ring has an efficient implementation on a quantum
computer. The second proof for the lemma \ref{char} and Section 4
suggest the character $\chi_\alpha$ (\ref{add_char}) as an example
for the {\em nice} additive character.

\section{An efficient implementation of $\calF_{R^\prime}$}

We study the factorization of $\calF_{R^\prime}$ (\ref{dftgalois})
in terms of $\calF_R=\calF_{R^\prime}|_{m=1}$ and then prove that it
can be efficiently performed on a quantum computer. We collect basic
facts on the discriminant matrix $D$ (\ref{discriminant}) of the
Galois ring $R^\prime$.

 \subsection{Factorization of $\calF_{R^\prime}$}

An $m\times m$  matrix associated with the basis
$\{\xi^i\}_{i=0}^{m-1}$ of the module $R^\prime$ on $R$, \eq
 \label{discriminant}
 D =(D_{ij})_{0\le i,j \le m-1}, \quad D_{ij}= Tr\,(\xi^{i+j}),
 \en
is called {\em the discriminant matrix} over the Galois ring
$R^\prime$, and it is the Hankel matrix satisfying
$D_{ij}=D_{i+1,j-1}$. We express the trace of the product of two
elements $x,y\in R^\prime$ as
 \eq
 Tr(x\cdot y)=\vec{x}^T D\vec{y} = \vec{x}^{\,\prime\, T} \vec{y}, \quad
 x_i^\prime=(D\vec{x})_i=Tr(x\xi^i).
 \en
Namely, we decomposes the trace of $x\cdot y$ as a linear summation
of the products of two elements $x_i^\prime, y_i\in R$.

Let $\calF_R$ denote the QFT  over the residue class ring of
integers $R$ ($\Z_{p^s}$),
 \eq
 \label{dftgal}
 \calF_R  \equiv \frac 1 {\sqrt{p^s}}
  \sum_{x_i,y_i\in R } (\omega_{p^s} )^{ x_i\cdot y_i} |y_i \rangle \langle x_i|,
   \quad 1 \le i \le m-1,
 \en
and then we describe $\calF_{R^\prime}$ (\ref{dftgalois}) as the
composition of an $m$-fold tensor product of $\calF_R$ and a shift
operator $\U_D$,
 \eq
 \label{dft}
 \calF_{R^\prime} = (\calF_R)^{\otimes m} \circ \U_D,
 \quad \U_D \equiv\sum_{x\in R^\prime} |x^\prime\rangle\langle x|,
 \en
where $x^\prime =(D\vec{x})^T\cdot \vec{\xi}$. Obviously, the
properties of  $\U_D$ are determined by the discriminant matrix $D$.

 \subsection{The discriminant matrix $D$ is invertible}

The set $\{|x\rangle\}_{x\in R^\prime}$ is an orthonormal basis of
the Hilbert space $(\C^{p^s})^{\otimes m}$. As the discriminant
matrix $D$ is invertible, the map $\vec{x}\mapsto
\vec{x}^{\,\prime}=D \vec{x}$ is bijective, and the set
$\{|x^\prime\rangle\}_{x^\prime\in R^\prime}$ also forms an
orthonormal basis of $(\C^{p^s})^{\otimes m}$. Hence the shift
operator $\U_D$ is a unitary transformation. Furthermore, we can
derive $\U_D^\dag =\U_{D^{-1}}$.

 \begin{lemma}
 \label{disinvert}
 The discriminant matrix $D$ (\ref{discriminant}) is invertible.
 \end{lemma}

 \begin{proof}
The discriminant matrix $D$ is invertible if and only if its rows
form a basis of $R^{1\times m}$, or equivalently, the following
equations
 \eq
 \sum_{i=0}^{m-1} b_i\,\cdot \, row_i(D) =0 \Leftrightarrow \sum_{i=0}^{m-1}
 b_i D_{ij} =0, \quad j=0, \cdots, m-1
 \en
 admit only $b_i =0$ as a solution.

 Assume a nonzero  solution $  \vec{b\,}^T=(b_0,\cdots,
 b_{m-1})$ of the equation $  \vec{b\,}^T \cdot D=0$ and a
 corresponding nonzero element $\beta=  \vec{b\,}^T \cdot
 \vec{\xi} \in R^\prime$. With another arbitrary
 element $\alpha= \vec{a}^T \cdot
   \vec{\xi} \in R^\prime$,
   we calculate
  \eq
  Tr(\beta \cdot \alpha) = \vec{b\,}^T \cdot D \cdot   \vec{a}=0
  \en
where $\vec{a}, \vec{b} \in R^{m\times 1}$. 1). As $\beta$ is a unit
of $R^\prime$, i.e., its inverse $\beta^{-1}$ exists, replacing
$\alpha$ with $\beta^{-1}\alpha$ gives rise to $Tr(\alpha)=0$ for
$\alpha\in R^\prime$. 2). As $\beta$ is a zero divisor of
$R^\prime$, with the lemma \ref{zdu}, it has a form of
 $\beta =p^k \beta^\prime$, $1\le k \le s-1$
where $\beta^\prime$ is a unit with the inverse
$(\beta^{\prime})^{-1}$, and then we have \eq
 p^k\, Tr (\beta^\prime \alpha) =0 \Rightarrow p^k\, Tr (\alpha) =0
 \en
suggesting that $Tr(\alpha)$ is either zero or a zero divisor of
$R$.

Hence, if $\beta\neq 0$, then $Tr(\alpha)$ for $\alpha\in R^\prime$
is a zero divisor or zero. This contradicts with the fact that the
trace map $Tr: R^\prime \to R$ is surjective. Therefore, $\beta=0$,
namely the equation $\vec{b\,}^T\cdot D=0$ only has a zero solution
$\vec{b}=0$,  which is equivalent to the existence of $D^{-1}$, the
inverse of the $D$ matrix.
 \end{proof}

This proof suggests: if $\{\xi^i\}^{m-1}_{i=0}$ is a basis of
$R^\prime$ then $D^{-1}$ exists. On the other hand, it is easy to
prove: if $D^{-1}$ exists then $\{\xi^i\}^{m-1}_{i=0}$ forms a basis
of $R^\prime$. The matrix $D$ is hence called the discriminant
matrix associated with the basis $\{\xi^i\}^{m-1}_{i=0}$ of the
Galois ring $R^\prime$.

\subsection{Remarks on the discriminant matrix $D$}

The lemma \ref{disinvert}, the existence of $D^{-1}$ over  the
Galois ring $R^\prime$, can be proved in the other way. If $D^{-1}$
exists, then the map $D: \vec{x} \mapsto D \vec{x}$ is bijective.
This means the kernel of this map $D$ is trivial, namely $D
\vec{x}=0$ if and only if $\vec{x}=0$. Assume a nonzero $\vec{y}$
satisfying  $D \vec{y}=0$. 1). As $y$ is a unit of $R^\prime$, the
set $\{y \xi^i \}_{i=0}^{m-1}$ is a new basis of the Galois ring
$R^\prime$. We expand $\alpha\in R^\prime$ with the new basis,
$\alpha =\sum_{i=0}^{m-1} a_i (y\xi^i)$, and then apply the trace
map to get $Tr(\alpha)=0$ for $\alpha\in R^\prime$ due to
$\vec{y\,}^T D=0$. 2). As $y$ is a zero divisor of $R^\prime$, with
the lemma \ref{zdu}, we denote $y=p^k y^\prime$, $1\le k \le s-1$
with $y^\prime$ a unit. We have  $p^k Tr(y^\prime \xi^i)=0$ due to
$Tr(y \xi^i)=0$. Expand $\alpha\in R^\prime$ with the new basis
$\{y^\prime \xi^i\}^{m-1}_{i=0}$, namely $\alpha=\sum_{i=0}^{m-1}
a_i^\prime (y^\prime \xi^i)$, and we have $p^k Tr (\alpha)=0$ which
suggests $Tr(\alpha)$ either a zero divisor or zero. Since the trace
map $Tr$ is surjective, the kernel of this map $D$ has to be
trivial, and hence $D$ is invertible.

Here, we make a sketch on how to compute the discriminant matrix $D$
(\ref{discriminant}) over the Galois ring $R^\prime$. Given a basic
primitive polynomial $\xi^m =\vec{h\,}^T\cdot \vec{\xi} $ from
(\ref{polynomial}) with roots $\xi$, $\xi^p$, $\cdots$,
$\xi^{p^{m-1}}$. 1). Compute $\xi^k$ in a recursive procedure, \eq
 \label{recursive}
 \xi^k = (\vec{h}^{(k-m)})^T\cdot \vec{\xi}, \quad \vec{h}^{(0)}=\vec{h},
 \quad
m \le k \le p^m-2 \en where $\vec{h}^{(k-m)}$ is calculated via
 \eq
\vec{h}^{(k-m)} = V^{k-m+1}
 \left(\begin{array}{c}
 0  \\
 0 \\
\vdots  \\
0 \\
1
\end{array}\right), \quad
V=\left(  \begin{array}{ccccc}
 0 & 1 & \cdots & 0 & h_0 \\
1 & 0 & \cdots & 0 & h_1 \\
0 & 1 & \cdots & 0 & h_2 \\
\vdots & \vdots & \ddots & \vdots & \vdots \\
0 & 0 & \cdots & 1 & h_{m-1}
\end{array}\right).
 \en
2). Compute $Tr(\xi^i)$, $1\le i \le m-1$ in terms of $\xi^j$, $1\le
j \le p^m-2$, with the definition of the trace (\ref{trace}). 3).
Compute $Tr(\xi^i)$, $m\le i \le 2 m-2$ in terms of $Tr(\xi^j)$,
$1\le j \le m-1$, with the help of the formula (\ref{recursive}).
4). We obtain all entries of the discriminant matrix $D$
(\ref{discriminant}).

\subsection{Complexity analysis of implementing $\calF_{R^\prime}$}

Denote $n=\log p^s$. Assume that the discriminator matrix $D$
(\ref{discriminant}) is known via relevant classical computation.

The factorization formalism (\ref{dft}) of the QFT
$\calF_{R^\prime}$ describes an efficient quantum circuit for the
implementation of $\calF_{R^\prime}$. It is known that the QFT
$\calF_R$ (\ref{dftgal}) can be efficiently approximated
\cite{hh00}. The invertible discriminant matrix $D$ gives rise to
the bijective map,
 $D: R^{\otimes m} \to R^{\otimes m}$. This map
can be efficiently performed as a permutation on a classical
computer, and hence the corresponding unitary transformation,
 $$\U_D: (\C^{p^s})^{\otimes m} \to (\C^{p^s})^{\otimes m}$$
 can be efficiently performed on a quantum computer \cite{ft82}.

The bijective map from $\vec{x}$ to $D\vec{x}$ ensures that the
vector $\vec{x}$ can be computed in a polynomial time with the known
$D$ and $D\vec{x}$.  The vector $D \vec{x}$ can be computed in time
$\CO(m^2)$. Hales and Hallgren \cite{hh00} proved that there exists
a quantum algorithm to approximate the QFT $\calF_R$ over
$R=\Z_{p^s}$ within accuracy $\epsilon$ which runs in time $\CO(n
\log \frac n \epsilon + \log^2 \frac 1 \epsilon)$. Hence
$\calF_{R^\prime}$ can be performed in a polynomial time $\CO(m^2) +
m \CO(n \log \frac n \epsilon + \log^2 \frac 1 \epsilon)$ within
accuracy $\epsilon$.

Let $C(p^s,\epsilon)$ denote the minimum size of a quantum circuit
approximating the QFT $\calF_R$ over $R$ within accuracy $\epsilon$,
and then performing $\calF_R^{\otimes m}$ needs a quantum circuit
with the size $m C(p^s,\epsilon)$. The matrix operation $D \vec{x}$
can be performed in a circuit with size $\CO(m^2 n^2)$, namely, each
arithmetic operation needs a circuit with size $n^2$. Hence
$\calF_{R^\prime}$ is performed on a quantum circuit with the size
$\CO(m^2 n^2) + m C(p^s,\epsilon)$.

Therefore, {\em the QFT  $\calF_{R^\prime}$ (\ref{dft}) over the
Galois ring $R^\prime$ can be performed within accuracy $\epsilon$
in a polynomial time
$$\CO(m^2)+m \CO \left( n \log \frac n \epsilon + \log^2 \frac 1
\epsilon\right )$$ and by a quantum circuit of the size $\CO(m^2
n^2) + m C(p^s,\epsilon)$.}

 \section{Comments on the QFT over a ring}

With the help of the QFT over Galois rings, the QFT over a finite
commutative ring with identity can be defined in principle.

A finite commutative ring with identity is expressed as a direct sum
of local rings, and a local commutative ring can be characterized as
a homomorphic image of a polynomial ring over a Galois ring, see
\cite{macdonald74, bf02} for related theorems and proofs. The
simplest example is the fundamental theorem of arithmetics: given a
unique prime factorization of the integer $m$ by
 $$m=p_1^{n_1} p_2^{n_2}\cdots p_k^{n_k}, \quad n_i \in \N, p_i \,\,
  \textrm{prime},
 \quad 1 \le i \le k,$$
there is a ring isomorphism
 $$\Z_m \cong \Z_{p_1^{n_1}} \oplus \Z_{p_2^{n_2}}\oplus\cdots
  \oplus \Z_{p_k^{n_k}} $$
which defines the QFT over $\Z_m$ in terms of the QFTs over
$\Z_{p^i}$.

De Beaudrap and coauthors \cite{bcw00} proved that if the QFT over a
ring has the property of the control/target inversion then the QFT
over the matrix ring has the same property. A matrix ring over a
finite commutative ring is often a noncommutative ring, and hence
the QFT over a noncommutative ring can be discussed via the QFT over
a finite commutative ring with identity.

 \section*{Acknowledgements}

The author thanks Hamed Ahmadi, Joseph Brennan, Daniel Nagaj, and
Martin R{\"o}tteler for relevant comments. Y. Zhang is in part
supported by NSF-China Grant-10605035.

 \end{document}